\documentclass[11pt]{amsart}

\usepackage[english]{babel}
\usepackage{amsmath,amssymb,amsthm}

\usepackage[T1]{fontenc}
\usepackage{caption}
\usepackage{times}
\usepackage{color,graphicx}
\usepackage{subfig,array}
\usepackage{amsthm,amssymb,amsmath,tensor}
\usepackage{hyperref}
\usepackage{tikz}
\usetikzlibrary{arrows,decorations.pathreplacing}

\tikzstyle{vertex}=[circle, draw, inner sep=0pt, minimum size=4pt]

\tikzset{every loop/.style={}}

\newcommand\tr{\mathop{\rm tr}\nolimits}

\def\Z{\mathbb{Z}}

\def\C{\mathbb{C}}

\def\R{\mathbb{R}}
\newcommand{\bes}{\begin{eqnarray*}}
	\newcommand{\ees}{\end{eqnarray*}}
\newcommand{\bpm}{\begin{pmatrix}}
	\newcommand{\epm}{\end{pmatrix}}

\def\diag{{\rm diag}\,}
\def\tr{{\rm tr}\,}

\def\modn{\,{\rm mod}\,}

\usepackage[numbers,sort&compress]{natbib}
\bibpunct[, ]{[}{]}{,}{n}{,}{,}
\makeatletter
\def\NAT@def@citea{\def\@citea{\NAT@separator}}
\makeatother

\theoremstyle{plain}
\newtheorem{theorem}{Theorem}[section]
\newtheorem{lemma}[theorem]{Lemma}
\newtheorem{corollary}[theorem]{Corollary}
\newtheorem{proposition}[theorem]{Proposition}

\theoremstyle{definition}

\newtheorem{example}[theorem]{Example}

\theoremstyle{remark}
\newtheorem{remark}{Remark}

\begin{document}


\title{Perfect quantum state transfer in weighted paths with potentials (loops) using orthogonal polynomials}

\author{Steve Kirkland\textsuperscript{a}}
\author{Darian McLaren\textsuperscript{b}}
\author{Rajesh Pereira\textsuperscript{c}}
\author{Sarah Plosker\textsuperscript{a, b, c}}\email{ploskers@brandonu.ca}
\author{Xiaohong Zhang\textsuperscript{a}}
\address{\textsuperscript{a}Department of Mathematics, University of Manitoba, Winnipeg, MB, Canada  R3T 2N2; \textsuperscript{b}Department of Mathematics \& Computer Science, Brandon University, Brandon,  MB, Canada R7A 6A9; \textsuperscript{c}Department of Mathematics \& Statistics, University of Guelph, Guelph, ON, Canada N1G 2W1}

\begin{abstract}
A simple method for transmitting quantum states within a quantum computer is via a quantum spin chain---that is, a path on $n$ vertices. Unweighted paths are of limited use, and so a natural generalization is to consider weighted paths; this has been further generalized to allow for loops (\emph{potentials} in the physics literature). We study the particularly important situation of perfect state transfer with respect to the corresponding adjacency matrix or Laplacian through the use of orthogonal polynomials. Low-dimensional examples are given in detail. Our main result is  that PST with respect to the Laplacian matrix cannot occur for weighted  paths on $n\geq 3$ vertices nor can it occur for certain symmetric weighted trees. The methods used lead us to a conjecture directly linking the rationality of the weights of weighted paths on $n>3$ vertices, with or without loops, with the capacity for PST between the end vertices with respect to the adjacency matrix.
\end{abstract}


\maketitle

\section{Introduction}
A quantum spin chain can be used as a means of accomplishing the important task of transferring a quantum state from one place to another within a quantum computer \cite{Bose}. From a graph theoretic perspective, we are interested in paths on $n$ vertices. It was found \cite{CDEL} that unweighted paths governed by XX dynamics (where one considers the corresponding adjacency matrix) only exhibit perfect state transfer (PST; a desirable property when transferring a quantum state) for $n\leq 3$.

Several offshoots have developed: one avenue is to consider pretty good state transfer (PGST; essentially this amounts to being ``arbitrarily close'' to PST). In this regard,  a complete characterization of the parameters (length of the unweighted chain)
for which there is PGST from one end vertex to the other  was given in \cite{GKSS12}, where it was shown that PGST occurs on an unweighted chain with $n$ nodes if and only if $n+1$ is either a prime number, two times a prime number, or a power of two.

A second avenue is to consider graphs other than paths to see if they also exhibit PST, and as it turns out, many different families of graphs have this property. Some examples are: a family of double-cone non-periodic graphs, certain joins of regular graphs with $K_2$ or with the empty graph on two vertices \cite{Ang09, Ang10}. Also, necessary and sufficient conditions for circulant graphs (Cayley graphs on the group $\mathbf{Z}_n$) to exhibit PST have been given in \cite{BMPDSDcir,B13}. The Cartesian product of two graphs with PST both at time $t$ has also been shown to exhibit PST \cite{AlEtAl}, in particular the $n$-fold Cartesian product of $K_2$ with itself (the $n$-cube) has PST between its antipodal vertices. For a more general family of graphs-- cubelike graphs -- it was shown that if the sum of all the elements in the connection set is not 0, then there is PST in that graph \cite{BGS08}, and when the sum is 0, a necessary and sufficient condition for such a graph to admit PST is given in \cite{cubelike}. The PST property of Hadamard diagonalizable graphs (a graph whose Laplacian is diagonalizable by a Hadamard  matrix) has also been  studied in \cite{HadaPST}; a simple eigenvalue characterization for such graphs to admit PST at time $\pi/2$ is given, which was used to construct more graphs with PST. However, a path is arguably the simplest graph structure, and since the graphs would need to be realized physically within a quantum computer, it is desirable to proceed with paths when possible, in order to minimize the  amount
of physical and technological resources required.

Along this line,  a third avenue is to consider weighted paths (we use the term ``path'' to mean an unweighted path). In this paper, we always consider weighted or unweighted paths on $n$ vertices, with vertex set $\{1, 2, \cdots, n\}$. In  \cite{CDDEKL}, it was shown that for XX dynamics, PST  can be achieved over arbitrarily long distances by allowing for different, but fixed, couplings
between the qubits on the chain (the strengths of couplings correspond to
the edge weights of the underlying weighted graph of the network; the edge weights they used to achieve PST from vertex 1 to vertex $n$  were $w(j, j+1)=w(j+1, j)=\sqrt{j(n-j)}$ for each $j\in \{1,\dots, n-1\}$). The case of other  weights, as well as the addition of potentials (represented mathematically as weighted loops in the graph) remained open. It was conjectured in \cite{CLMS}, based on numerical evidence, that paths of arbitrary length $n$ can be made to have PST from vertex 1 to vertex $n$ by the addition of a suitable amount of energy (such \emph{energy shifts} are later \cite{KemptonPST, KemptonPGST} referred to as a potential function on the vertex set, or simply as potentials). This conjecture was then raised as an open problem in \cite{Godsil2011}.  Asymptotic \cite{asym} (taking the potentials at the endpoints arbitrarily large) and approximate \cite{approx} results gave affirmative answers to the conjecture in the respective settings.
Very recently, the conjecture was shown to be false for the PST setting \cite{KemptonPST} but true in the more relaxed setting of PGST \cite{KemptonPGST}.

We consider this third avenue from a matrix analysis point of view: weighted paths and weighted paths with potentials (loops)  amount to tridiagonal matrices with certain restrictions (e.g.\ the diagonal entries are necessarily zero for weighted or unweighted paths without potentials when considering XX dynamics; that is, when considering the adjacency matrix associated to the graph). Any symmetric tridiagonal matrix gives way to a three-term recurrence relation, and so our approach is to work with the orthogonal polynomials that arise by considering the tridiagonal matrix as an operator on the polynomial space.

Given the eigenvalues of a weighted path with or without loops satisfying certain conditions, several  algorithms exist  for constructing a tridiagonal matrix corresponding to a graph with PST: in \cite{dBG}, the authors produce formulas for calculating the weights of the discrete inner product that arises through the orthogonal polynomials; in \cite[Chapter 4]{Gladwell} two methods for computing the eigenvectors needed in the inversion procedure are reviewed, in addition to the method in \cite{dBG}, and in \cite{vinet}, the authors use the Euclidean algorithm starting with $p_n$ and $p_{n-1}$ in order to find $p_{n-2}$ and then repeat until all the orthogonal polynomials have been found.

 Using an approach similar to that found in \cite{orthogpolyGTLZ}, we obtain formulas, in terms of the eigenvalues, for the weight of the edge between vertices $\lfloor\frac{n}2\rfloor$ and $\lfloor\frac{n}2+1\rfloor$ and the potential (if allowing for loops) at vertex $\lceil\frac{n}2\rceil$. This allows one to determine the ``middle'' weights of the weighted path (with or without loops) without the need to calculate all the orthogonal polynomials or the weights of the inner product. We give examples for $n<5$ illustrating the utility of our formulas.
 We then show that if we consider XXX dynamics (that is, we consider the Laplacian associated to the graph), a weighted path on at least three vertices cannot have PST. We extend this result to apply to symmetric weighted trees.  For  XX dynamics, our analysis leads us to propose the following conjecture:  weighted paths on at least four vertices with or without loops must have at least one irrational weight in order to have adjacency matrix PST at a fixed readout time $\pi$; we confirm  this conjecture for $n=4$ as well as for $n\equiv 3\modn 8$ and for $n\equiv 5\modn 8$.
These results shed light on the nature of weighted  paths with or without loops that exhibit PST between the endpoints.

Throughout this paper,
in the setting of adjacency matrices, we consider paths with or without loops. In the Laplacian matrix setting, we only consider paths without loops.


The paper is organized as follows. In Section \ref{sec:prelim} we review preliminary information necessary for our results. In Section \ref{sec:alg}, we give our formulas for constructing the middle weights of the weighted path and provide examples for small $n$ ($n=2$ and $n=3$) for both the adjacency matrix and Laplacian matrix settings. In Section \ref{sec:thmL}, we prove that a weighted path  with at least three vertices cannot have Laplacian PST and generalize the result to mirror symmetric trees. In Section \ref{sec:conjrat}, we consider further examples in the adjacency matrix setting ($n=4$ and $n=5$) and  prove a negative result for rational-weighted paths with or without potentials with respect to adjacency matrix PST; we conjecture that this holds for all $n\geq 4$.

\section{Preliminaries}\label{sec:prelim}

Given a quantum system, the Hamiltonian $H$ is a matrix representing the total energy of the system; its spectrum represents the possible measurement outcomes when one measures the total energy. The  dynamics of the system lead us to consider either the adjacency matrix or the Laplacian of the graph corresponding to the system.

Let $G$ be an undirected graph on $n$ vertices ($G$ here can be either weighted or unweighted). The corresponding \emph{adjacency matrix} is an $n\times n$ matrix $A=[a_{jk}]$ whose entries satisfy $a_{jk}=  w(j,k)$, where $w(j,k)$ is the weight of the edge between vertex $j$ and vertex $k$ (if there is no edge between the two vertices, then $w(j,k)=0$; if the graph is unweighted then all edges are taken to have weight 1).  The \emph{degree} of a vertex is the sum of the weights of the incident edges, and we can create a diagonal degree matrix $D$ whose $(j,j)$ entry is the degree of vertex $j$. The \emph{Laplacian matrix} (or simply, the Laplacian) corresponding to a simple graph $G$ is $L=D-A$, which is a positive semidefinite matrix with smallest eigenvalue zero (its multiplicity is equal to the number of connected components of the graph---and thus equal to one herein---with the all-ones vector as its eigenvector).

We are interested in the matrix exponential $e^{itH}$, where $H$ is the Hamiltonian of the system (i.e., the adjacency matrix or the Laplacian matrix, depending on the dynamics) and $t$ is the readout time. Let $1\leq j<k\leq n$. There is \emph{perfect state transfer} from vertex $j$ to vertex $k$ if there exists some time $t=t_0$ such that $|e_j^T e^{it_0H} e_k|^2=1$, where $\{e_\ell\}_{\ell=1}^n$ is the standard ordered basis. There is \emph{pretty good state transfer} from vertex $j$ to vertex $k$ if for any $\epsilon>0$, there exists a time $t=t_\epsilon$ such that $|e_j^T e^{it_\epsilon H} e_k|^2>1-\epsilon$. Note that, because $G$ is undirected, PST (or PGST) occurs from vertex $j$ to vertex $k$ if and only if it occurs from vertex $k$ to vertex $j$.

Here, we focus on two settings: weighted or unweighted paths governed by XXX dynamics; and weighted or unweighted paths, that may or may not   have loops, governed by XX dynamics. In both settings, the vertices of the path are labelled so that vertex $j$ is adjacent to vertex $j+1, j=1, \ldots, n-1.$ As a result,  our Hamiltonian will always be a tridiagonal matrix of one of the following two forms depending on the dynamics ($A$ for XX dynamics, $L$ for XXX dynamics)
\begin{eqnarray}\label{eq:H}
A=\begin{bmatrix}
q_1&r_1&&&&&\\
r_1&q_2&r_2&&&&\\
&r_2&q_3&r_3&&&\\
&&&\ddots&&&\\
&&&&&r_{n-1}\\
&&&&r_{n-1}&q_n
\end{bmatrix},
L=\begin{bmatrix}
q_1&-r_1&&&&&\\
-r_1&q_2&-r_2&&&&\\
&-r_2&q_3&-r_3&&&\\
&&&\ddots&&&\\
&&&&&-r_{n-1}\\
&&&&-r_{n-1}&q_n
\end{bmatrix}
\end{eqnarray}
For the adjacency matrix case, all  $q_j=0$ for  (unweighted or weighted) paths without potentials. A potential at vertex $j$ corresponds to an entry $q_j\neq 0$. For the Laplacian matrix case, $q_1=r_1$, $q_j=r_{j-1}+r_j$ for $j=2, 3, \dots, n-1$, and $q_n=r_{n-1}$. In both cases, $r_j>0$ denotes the weight of the edge between vertex $j$ and $j+1$.

As in \cite{orthogpolyGTLZ}, we note that both of the above two symmetric tridiagonal $n\times n$ matrices are connected to a set of $n$ orthogonal polynomials via the three term recurrence given by
\begin{eqnarray}\label{eq:3termrec}
p_k(x)=(x-q_k)p_{k-1}(x)-r_{k-1}^2p_{k-2}(x) \quad\textnormal{ for all $k\in \{1, \dots, n\}$}
\end{eqnarray}  where we define $p_{-1}(x)=0$ and $p_0(x)=1$. We will denote the eigenvalues of $A$ or $L$ (the roots of $p_n(x)$) by  $\alpha_r$.

Rearranging equation (\ref{eq:3termrec}), we find that $xp_{k-1}(x)=p_k(x)+q_kp_{k-1}(x)+r_{k-1}^2p_{k-2}(x)$, and thus we can consider the operator
\begin{eqnarray}\label{eq:opM}
M=\begin{bmatrix}
q_1&1&&&&\\
r_1^2&q_2&1&&&\\
&r_2^2&q_3&1&&\\
&&\ddots&\ddots&\ddots&\\
&&&r_{n-2}^2&q_{n-1}&1\\
&&&&r_{n-1}^2&q_n
\end{bmatrix}.
\end{eqnarray}
The matrix $M$ represents multiplication by $x$ ($\modn p_n(x)$) in the basis $\{p_0(x),\dots, p_{n-1}(x)\}$.   We note that $A$ is similar to $M$ via $QM=AQ$ where $Q=\diag(d_1, \dots, d_n)$ and
\begin{eqnarray}\label{eqndone}
d_j=\begin{cases}
      \displaystyle\frac{1}{\prod_{\ell=1}^{j-1} r_\ell} & \textnormal{ if }j\neq 1 \\
      1 &\textnormal{ if }j=1
   \end{cases}
\end{eqnarray}

 The Matrix $L$ is also similar to $M$ via $TM=LT$ where $T=\diag(d_1, \dots, d_n)$ and
\begin{eqnarray}\label{eqndtwo}
d_j=\begin{cases}
      \displaystyle\frac{(-1)^{j-1}}{\prod_{\ell=1}^{j-1} r_\ell} & \textnormal{ if }j\neq 1 \\
      1 &\textnormal{ if }j=1
   \end{cases}
\end{eqnarray}

Letting $H$ denote the matrix $A$ or $L$ in equation (\ref{eq:H}), we also note that the eigenvalues of $H$ are real and  distinct since  $r_j\neq 0$ for all $j$ (see, e.g.\ \cite[Chapter 4]{Atkinson}); this allows for the multiplication by $x$ viewpoint to hold.
 We then use this distinctness to order the eigenvalues as follows:
 \begin{eqnarray}\label{eq:ordalpha}
 \alpha_1<\alpha_2<\cdots<\alpha_n
 \end{eqnarray}
Note that the eigenvector of $M$ associated to the eigenvalue $\alpha_j$ is $w_j=[p_0(\alpha_j),\;p_1(\alpha_j),\; \cdots,\; p_{n-1}(\alpha_j)]^T$. This can be verified by computing $Mw_j$, and then using the  recurrence relation (\ref{eq:3termrec}) evaluated at $\alpha_j$ to simplify each term.

Now, let us consider the set of polynomials $S=\{\tilde{p}_0(x),\dots,\tilde{p}_{n-1}(x)\}$ with $\tilde{p}_k(x)=d_{k+1} p_k(x)$, where the $d$'s are given by Equation (\ref{eqndone}) if we are taking $H=A$ and the $d$'s are given by Equation (\ref{eqndtwo}) if we are taking $H=L$. The set $S$ is a basis of the vector space of all polynomials of degree less than $n$.  In the basis $S$, the matrix that represents the multiplication by the $x$ ($\modn p_n(x)$) operator  is exactly $H$.


Let $H$ denote the matrix $A$ or $L$ in Equation (\ref{eq:H}). Let $v_r=[\tilde{p}_0(\alpha_r),\tilde{p}_1(\alpha_r),\ldots,\tilde{p}_{n-1}(\alpha_r)]^T$, then the two vectors $v_r$ and $v_s$ are orthogonal to each other for any $r\neq s$. Normalizing these vectors, assume the factors are $\sqrt{\kappa_j}$, $j=1,\ldots, n$, respectively, then the matrix $V=[\sqrt{\kappa_1}v_1,\sqrt{\kappa_2}v_2,\ldots, \sqrt{\kappa_n}v_n]$ is an orthogonal matrix, and it diagonalizes the Hamiltonian $H$ to the diagonal matrix $\Lambda=\diag(\alpha_1,\ldots,\alpha_n)$, i.e., $V^THV=\Lambda$. If there is PST between vertex $j+1$ and $k+1$ at time $t=t_0$, then $1=|e_{j+1}^Te^{it_0H}e_{k+1}|=|e_{j+1}^TVe^{it_0\Lambda}V^Te_{k+1}|=|r_{j+1}e^{it_0\Lambda}r^T_{k+1}|$, where $r_j$ is the $j$-th row of the matrix $V$. From the Cauchy-Schwarz inequality, we know $r_{j+1}e^{it_0\Lambda}=e^{i\phi}r_{k+1}$ for some phase factor $\phi$, which can be rewritten as 
\begin{eqnarray}
\frac{\tilde{p}_k(\alpha_r)}{\tilde{p}_j(\alpha_r)}=e^{-i\phi}e^{it_0\alpha_r}
\end{eqnarray}
for $r=1,2,\ldots, n$,  and some phase factor $\phi$. Since the polynomials $\tilde{p}_j(x)$ are real, it follows that $\tilde{p}_k(\alpha_r)/\tilde{p}_j(\alpha_r)=\pm 1$.
In fact, if we consider the case of PST between the endpoints, it is well known \cite{vinet} that $\tilde{p}_{n-1}(\alpha_r)=(-1)^{n+r}$. Looking at two neighbouring eigenvalues $\alpha_r$ and $\alpha_{r-1}$,  Kay \cite{Kay2010} found that
$\alpha_r-\alpha_{r-1}=(2m_r+1)\pi/t_0$ where $m_r$ is any nonnegative integer.
Here we scale the Hamiltonian ($A$ or $L$ depending on the dynamics) by a factor $t_0/\pi$ so that the PST time is $\pi$, and we therefore look at the simpler expression \begin{eqnarray}\label{eq:Kayodd}
\alpha_r-\alpha_{r-1}=(2m_r+1).
\end{eqnarray}

It is known \cite[Lemma 2]{Kay2010} 
 that for a symmetric tridiagonal Hamiltonian $H$, if PST occurs between the end vertices, then $H$ must also be persymmetric (symmetric about the anti-diagonal; such persymmetric matrices are also called mirror symmetric in the literature). In the case of a weighted path having no potentials and  governed by XX dynamics (therefore $q_1, \dots, q_n$ are all zeros),
the associated graph is then bipartite, and by properties of bipartite graphs the eigenvalues are symmetric about zero.
In this case, we give the eigenvalues another set of labels as follows
\begin{eqnarray}\label{eq:alphas}
-\beta_{\frac{n}2}<\cdots<-\beta_2<-\beta_1<0<\beta_1<\beta_2<\cdots<\beta_{\frac{n}2}
\textnormal{, for $n$ even }\nonumber \\
-\beta_{\frac{n-1}2}<\cdots<-\beta_2<-\beta_1<\beta_0=0<\beta_1<\beta_2<\cdots<\beta_{\frac{n-1}2} \textnormal{, for $n$ odd }
 \end{eqnarray}
  (we use zero as the index of the zero eigenvalue in the  case that $n$ is odd; zero does not appear as an eigenvalue in the  case that $n$ is even). From now on, when we mention the eigenvalues as $\alpha_r$, we mean the ones ordered as in (\ref{eq:ordalpha}); and when we mention eigenvalues $\beta_r$ we mean the ones as in (\ref{eq:alphas}).
  If $n$ is even then  (\ref{eq:Kayodd}) and (\ref{eq:alphas}) yield the fact that  $\beta_1-(-\beta_1)=2\beta_1=(2m_1+1)$, and therefore $\beta_1=(2m_1+1)/2$. Using this, we find
\begin{eqnarray}
\beta_2-\beta_1&=&(2m_2+1)\\
\Rightarrow \beta_2&=&(2m_2+1)+\frac{(2m_1+1)}2\\
&=&\frac{(4m_2+2m_1+3)}2.
\end{eqnarray}
Following this, we see that if $n$ is even, all $\beta_r$ will be odd multiples of $1/2$. In fact, one can easily show by continuing the analysis of $\beta_r-\beta_{r-1}$, that the $\beta_r$ alternate between $1\modn 4$  times $1/2$ and $3\modn 4$ times $1/2$ (these give us alternating $\pm i$ when considering   $e^{i\pi\beta_r}$ in the matrix exponential $e^{i\pi H}$). A similar analysis shows that if $n$ is odd, the $\beta_r$ are even multiples of $1/2$, alternating between  $0\modn 4$  times $1/2$ and $2\modn 4$ times $1/2$ (these give us alternating $\pm 1$ when considering $e^{i\pi\beta_r}$), with $\beta_0=0\equiv 0\mod 4$.
We summarize this in the following remark:

\begin{remark}\label{rmk:multpi2}
For the adjacency matrix of a weighted path  without potentials that exhibits PST between the end vertices at time $\pi$, the eigenvalues $\beta_r$ adhere to the following pattern: for $n$ even, the $\beta_r$ alternate between $(1\modn 4)\times 1/2$ and $(3\modn 4)\times 1/2$, while for $n$ odd, the $\beta_r$ alternate between  $(0\modn 4)\times 1/2$ and $(2\modn 4)\times 1/2$.
\end{remark}


For the adjacency matrix of a weighted path with loops, we can shift all the eigenvalues (by adding a multiple of the identity) such that the smallest eigenvalue is an integer; equation (\ref{eq:Kayodd}) then tells us they must alternate even and odd. This new weighted path (possibly with potentials) will exhibit PST if and only if the original path does. The eigenvalues can then be assumed to be integers with alternating parity without loss of generality. In the case of XXX dynamics, $L$ is positive semi-definite with smallest eigenvalue 0 (with multiplicity 1 since the graph is connected). Using this together with equation (\ref{eq:Kayodd}), we know the integer sequence of ordered eigenvalues $\alpha_r$ begins with 0 (even number) and then alternates odd, even, odd, ... for all the remaining eigenvalues. We summarize this in the following remark:

\begin{remark}\label{rmk:multpi}
If a weighted path with potentials exhibits PST at time $\pi$, then the eigenvalues $\alpha_r$ of its adjacency matrix can be taken to alternate between even and odd (or odd and even) integers. Without loss of generality, for notational simplicity, we can shift so that the odd-indexed eigenvalues are odd, and the even-indexed eigenvalues are even (so $\alpha_1, \dots, \alpha_n$ alternate between odd and even). If the Laplacian of a weighted path with no potentials exhibits PST at time $\pi$, then the eigenvalues $\alpha_r$ alternate between even and odd integers (starting with the smallest eigenvalue: zero).
\end{remark}

\begin{proposition}\label{prop:internal}
For a weighted path with or without potentials, PST between vertices $1$ and $n$ implies PST between vertices $j$ and $n+1-j$, for each $j=2, \ldots, n-1$.
If for some $j$ with $2 \le j \le n-1$ there is PST between vertices $j$ and $n+1-j$, and if in addition none of the eigenvectors of the Hamiltonian has a zero entry in the $j$--th position, then the converse holds.
\end{proposition}

It is an open question if the converse holds even if there are eigenvectors with zeros in the jth place.  We note that in the more relaxed setting of PGST, this more general conjectured version of the converse fails to be true \cite{internal}: PGST can occur between internal vertices of paths in the absence of PGST between the end vertices.

\begin{proof}
Consider the matrix $M$ in equation (\ref{eq:opM}). Recall that the eigenvector $w_r$ of this matrix corresponding to the eigenvalue $\alpha_r$ is
\begin{eqnarray}
{w}_r=\begin{bmatrix}
p_0(\alpha_r)\\
\vdots\\
p_{n-1}(\alpha_r)
\end{bmatrix}.
\end{eqnarray}
Since $H=QMQ^{-1}$ (for $H=A$)  or $H=TMT^{-1}$ (for $H=L$), the eigenvector ${v}_r$ for $H$ corresponding to the eigenvalue $\alpha_r$ is
\begin{eqnarray}
{v}_r&=&Q{w}_r \quad \textnormal{ (or ${v}_r=Tw_r$)}\nonumber\\
&=&\begin{bmatrix}
d_1p_0(\alpha_r)\\
\vdots\\
d_np_{n-1}(\alpha_r)
\end{bmatrix}.
\end{eqnarray}
Now, if we assume that there is PST between the endpoints then $H$ must be mirror symmetric. The eigenvectors will therefore be either symmetric or antisymmetric (i.e $(v_r)_j=\pm(v_r)_{n-j+1}$, $j=1,2,\dots,n$) \cite[Theorem 2]{CantoniButler}, and recall that $\tilde{p}_j(x)=d_{j+1}p_j(x)$, we have either $\tilde{p}_{j-1}(\alpha_r)$ and $\tilde{p}_{n-j}(\alpha_r)$ are both zero, or neither of them is zero and for some phase factor $\hat{\phi}$ they satisfy
\begin{eqnarray}\label{eq:alternate pm1}
\frac{\tilde{p}_{j-1}(\alpha_r)}{\tilde{p}_{n-j}(\alpha_r)}=
\frac{\tilde{p}_0(\alpha_r)}{\tilde{p}_{n-1}(\alpha_r)}=\pm 1=e^{i(\pi\alpha_r-\hat{\phi})}.
\end{eqnarray}
The above is valid for all $\alpha_r$ and $j$ such that $\tilde{p}_{n-j}(\alpha_r)\neq 0$, and the quotients share the same alternating pattern between 1 and $-1$ determined by the PST between the end vertices; hence there is perfect state transfer between the vertices $j$ and $n+1-j$. 

The steps above are all reversible under certain conditions: if there is PST between a pair of  inner vertices $j$  and $n+1-j$ for some $2\leq j\leq n-1$, and if  $p_{j-1}(\alpha_r)\neq 0$ for all $\alpha_1,\,\cdots,\,\alpha_n$ (and therefore $p_{n-j}(\alpha_r)\neq 0$ as well), then equation (\ref{eq:alternate pm1}) is true for all $\alpha_r$ and the given $j$, and therefore there is PST between the two end vertices. 
\end{proof}

Referring to the proof of Proposition \ref{prop:internal}, we observe in passing that if $p_{j-1}(\alpha_r)=p_{n-j}(\alpha_r)=0$ for some $r$, then although the eigenvector symmetry/antisymmetry condition still holds,  it does not provide the $e^{i(\pi\alpha_r-\hat{\phi})}=\pm1$ constraints on the eigenvalues needed for PST between end vertices.

\section{Constructing matrices guaranteed to have PST for weighted paths with or without loops}\label{sec:alg}
Given a set of eigenvalues (with restrictions given from equation (\ref{eq:Kayodd})), we would like to reconstruct  the adjacency matrix of a weighted path with or without potentials, 
that is guaranteed to have PST between vertices $1$ and $n$.
That is, by choosing values for $\alpha_1,\dots, \alpha_n$ satisfying  Equation (\ref{eq:Kayodd}) (these will correspond to the eigenvalues of the adjacency matrix), one can reverse-engineer weighted paths, with or without potentials, having PST.
We go through the low-dimensional cases in detail in this section and the next section.

We next state a technical result that is especially helpful in analyzing the eigenvalues of matrices that are persymmetric.

\begin{lemma}\cite[Lemma 3]{CantoniButler}\label{lem:CB}
\begin{enumerate} Let $R$ be the \emph{reversal matrix}: an antidiagonal matrix with all ones along the antidiagonal.
\item If $n$ is even, the matrices $B=\begin{bmatrix}
E&RCR\\
C& RER
\end{bmatrix}$ and $\begin{bmatrix}
E-RC&0\\0&E+RC
\end{bmatrix}$ are orthogonally similar, where $C$ is some $\frac{n}2\times \frac{n}2$ matrix, and $R$ is also $\frac{n}2\times \frac{n}2$.
\item If $n$ is odd, the matrices $B=\begin{bmatrix}
E&x&RCR\\x^T&q&x^TR\\C&Rx&RER
\end{bmatrix}$ and $\begin{bmatrix}
E-RC&0&0\\0&q&\sqrt{2}x^T\\0&\sqrt2x& E+RC
\end{bmatrix}$ are orthogonally similar, where $C$ is some $\frac{n-1}2\times \frac{n-1}2$ matrix, $R$ is also $\frac{n-1}2\times \frac{n-1}2$, $q\in \mathbb{R}$, and $x\in \mathbb{R}^{\frac{n-1}2}$.
\end{enumerate}
\end{lemma}

Let $S_1=\sum_{r=1}^n(-1)^{r+n}\alpha_r$ and $S_2=\sum_{r=1}^n(-1)^{r+n}\alpha_r^2$.

\begin{corollary}\label{cor:middle} Let $A$ be the Hamiltonian  of a weighted mirror-symmetric path governed by XX dynamics, with or without loops, on $n$ vertices.
If $n$ is even then $\displaystyle r_{\frac{n}2}=\frac{S_1}2$ and $\displaystyle q_{\frac{n}2}=\frac{S_2}{2S_1}$. If $n$ is odd then  $\displaystyle r_{\frac{n-1}2}=\frac{\sqrt{S_2-S_1^2}}2$ and $q_{\frac{n+1}2}=S_1$.
\end{corollary}

\begin{proof}
We use the notation of Lemma \ref{lem:CB}. Suppose $n$ is even. Our Hamiltonian
\begin{eqnarray}
\begin{bmatrix}
q_1&r_1&&&&&&&&\\
r_1&q_2&r_2&&&&&&&\\
&r_2&q_3&r_3&&&&&&\\
&&&\ddots&&&&&&\\
&&&r_{\frac{n}2-1}&q_{\frac{n}2}&r_{\frac{n}2}&&&&\\
&&&&r_{\frac{n}2}&q_{\frac{n}2}&r_{\frac{n}2-1}&&&\\
&&&&&&\ddots&&&\\
\\
&&&&&&&r_2&q_2&r_1\\
&&&&&&&&r_1&q_1
\end{bmatrix}
\end{eqnarray}
is thus orthogonally similar to a 2-by-2 block diagonal matrix with diagonal blocks $$B_1=\begin{bmatrix}
q_1&r_1&&&&\\
r_1&q_2&r_2&&&\\
&r_2&q_3&r_3&&\\
&&&\ddots&&\\
&&&r_{\frac{n}2-1}&(q_{\frac{n}2}-r_{\frac{n}2})
\end{bmatrix} \textrm{ and } B_2=\begin{bmatrix}
q_1&r_1&&&&\\
r_1&q_2&r_2&&&\\
&r_2&q_3&r_3&&\\
&&&\ddots&&\\
&&&r_{\frac{n}2-1}&(q_{\frac{n}2}+r_{\frac{n}2})
\end{bmatrix}.$$ Here, $C$ has $r_{\frac{n}2}$ in its $(1,n/2)$ entry and zeros everywhere else.
Note that $B_2=B_1+2r_{\frac{n}2}e_{\frac{n}2}e_{\frac{n}2}^T$. It is a well-known fact that if one perturbs a Hermitian matrix by a rank-one symmetric matrix, the original matrix and the perturbed matrix will have interlacing eigenvalues. Since $r_{\frac n 2}$ is positive, 
it follows that $B_1$ has eigenvalues $\alpha_1, \alpha_3, \dots, \alpha_{n-1}$ and $B_2$ has eigenvalues $\alpha_2, \alpha_4, \dots, \alpha_n$. 

From the fact that the trace of a matrix is the sum of its eigenvalues, we find  $2r_{\frac{n}2}=\tr(B_2)-\tr(B_1)=S_1$ and therefore $\displaystyle r_{\frac{n}2}=\frac{S_1}2$. Now, from the fact that the trace of the square of a matrix is the sum of the squares of the eigenvalues of the original matrix, we find
\begin{eqnarray}
(q_{\frac{n}2}+r_{\frac{n}2})^2-(q_{\frac{n}2}-r_{\frac{n}2})^2&=&\tr(B_2^2)-\tr(B_1^2)=S_2\\
\Rightarrow 4q_{\frac{n}2}r_{\frac{n}2}&=&S_2\\
\Rightarrow q_{\frac{n}2}&=&\frac{S_2}{2S_1}.
\end{eqnarray}

Suppose $n$ is odd. Our Hamiltonian
\begin{eqnarray}
\begin{bmatrix}
q_1&r_1&&&&&&&&&\\
r_1&q_2&r_2&&&&&&&&\\
&r_2&q_3&r_3&&&&&&&\\
&&&\ddots&&&&&&&\\
&&&r_{\frac{n-1}2-1}&q_{\frac{n-1}2}&r_{\frac{n-1}2}&&&&&\\
&&&&r_{\frac{n-1}2}&q_{\frac{n+1}2}&r_{\frac{n-1}2}&&&&\\
&&&&&r_{\frac{n-1}2}&q_{\frac{n-1}2}&r_{\frac{n-1}2-1}&&&\\
&&&&&&&\ddots&&&\\
\\
&&&&&&&&r_2&q_2&r_1\\
&&&&&&&&&r_1&q_1
\end{bmatrix}
\end{eqnarray}
is orthogonally similar to a block diagonal matrix with diagonal blocks \\
$B_1=\begin{bmatrix}
q_1&r_1&&&&\\
r_1&q_2&r_2&&&\\
&r_2&q_3&r_3&&\\
&&&\ddots&&\\
&&&r_{\frac{n-1}2-1}&q_{\frac{n-1}2}
\end{bmatrix}$ and  $B_2=\begin{bmatrix}
q_{\frac{n+1}2}&0&\cdots&&&\sqrt2r_{\frac{n-1}2}\\
0&&&&&\\
\vdots&&&&&\\
&&&B_1&&\\
\sqrt2r_{\frac{n-1}2}&&&&
\end{bmatrix}$. Here, $C$ is the zero matrix.
From Cauchy's interlacing theorem for a bordered Hermitian matrix,
we know the eigenvalues of $B_1$ are $\alpha_2, \alpha_4, \dots, \alpha_{n-1}$, and the eigenvalues of $B_2$ are $\alpha_1, \alpha_3, \dots, \alpha_n$. A trace argument similar to the even case yields $q_{\frac{n+1}2}=\tr(B_2)-\tr(B_1)=S_1$ and $q_{\frac{n+1}2}^2+4r_{\frac{n-1}2}^2=\tr(B_2^2)-\tr(B_1^2)=S_2\Rightarrow \displaystyle r_{\frac{n-1}2}=\frac{\sqrt{S_2-S_1^2}}2$.
\end{proof}

\begin{remark}\label{Lapre}
For a weighted persymmetric path on $n$ vertices governed by XXX dynamics we have a similar result for the Hamiltonian $L$:
if $n$ is even, then $\displaystyle r_{\frac{n}2}=\frac{S_1}2$ and $\displaystyle q_{\frac{n}2}=\frac{S_2}{2S_1}$; if $n$ is odd, then  $\displaystyle r_{\frac{n-1}2}=\frac{\sqrt{S_2-S_1^2}}2$ and $q_{\frac{n+1}2}=S_1$. Furthermore, for the even case, from $q_{\frac n 2}=r_{\frac n 2-1}+r_{\frac n 2}$, we have $r_{\frac n 2-1}=\frac{S_2-S_1^2}{2S_1}$. While for the odd case, from $q_{\frac{n+1}{2}}=2r_{\frac{n-1}{2}}$, we have $S_2=2S_1^2$.

\end{remark}

\begin{example}[$2\times 2$ and $3\times 3$ Cases]

For $n=2$,  Corollary \ref{cor:middle} yields a weighted path with potentials having $\displaystyle r_1=\frac{\alpha_2-\alpha_1}2$ and $\displaystyle q_1=\frac{\alpha_2+\alpha_1}2$. 
If we consider a weighted path with no potentials (and so $\alpha_1, \alpha_2$ are simply $-\beta_1$, $\beta_1$),  the Hamiltonian $A$  reduces to $\begin{bmatrix}
0&\beta_1\\
\beta_1&0
\end{bmatrix}$, which shows that the unweighted path on $2$ vertices, given by the adjacency matrix $\begin{bmatrix}
0&1\\
1&0
\end{bmatrix}$ has PST from vertex 1 to vertex 2.
Similarly, for the Laplacian case, since $\alpha_1=0$, the Hamiltonian $L$ reduces to
$\frac12\begin{bmatrix}
\alpha_2&-\alpha_2\\-\alpha_2&\alpha_2
\end{bmatrix}$.

For $n=3$, Corollary \ref{cor:middle} yields a weighted path with potentials governed by the XX dynamics having $\displaystyle r_1=\frac{\sqrt{-2\alpha_2^2+2\alpha_1\alpha_2-
2\alpha_1\alpha_3+2\alpha_2\alpha_3}}2$ and $\displaystyle q_2=\alpha_1-\alpha_2+\alpha_3$. Here $B_1$ is simply the $1\times 1$ matrix $(q_1)$, and thus we can find $q_1$ via $q_1=\tr(B_1)=\alpha_2$.

Under XX dynamics, the Hamiltonian $A$ is
\begin{eqnarray*}
A=\begin{bmatrix}
\alpha_2&\frac{\sqrt{-2\alpha_2^2+2\alpha_1\alpha_2-
2\alpha_1\alpha_3+2\alpha_2\alpha_3}}2&0\\\frac{\sqrt{-2\alpha_2^2+2\alpha_1\alpha_2-
2\alpha_1\alpha_3+2\alpha_2\alpha_3}}2&\alpha_1-\alpha_2+\alpha_3&\frac{\sqrt{-2\alpha_2^2+2\alpha_1\alpha_2-
2\alpha_1\alpha_3+2\alpha_2\alpha_3}}2\\
0&\frac{\sqrt{-2\alpha_2^2+2\alpha_1\alpha_2-
2\alpha_1\alpha_3+2\alpha_2\alpha_3}}2&\alpha_2
\end{bmatrix}.
\end{eqnarray*} If we consider a weighted path with no potentials (and so $\alpha_1$, $\alpha_2$, $\alpha_3$ are simply $-\beta_1$, $0$, $\beta_1$),  the adjacency matrix reduces to
\begin{eqnarray}\label{eq:H3}
A=\begin{bmatrix}
0&\frac{\beta_1}{\sqrt{2}}&0\\
\frac{\beta_1}{\sqrt{2}}&0&\frac{\beta_1}{\sqrt{2}}\\
0&\frac{\beta_1}{\sqrt{2}}&0\\
\end{bmatrix}.
\end{eqnarray}

Our $3\times 3$ example is consistent with the literature: the unweighted path on 3 vertices given by the adjacency matrix $A=\begin{bmatrix}
0&1&0\\1&0&1\\0&1&0
\end{bmatrix}
$ has PST from vertex 1 to vertex 3, and from the above we see that the weighted case is simply a scalar multiple ($\beta_1/\sqrt{2}$) of the unweighted case.

Similarly for the XXX dynamics, from Remark \ref{Lapre} and the fact that $\alpha_1=0$, we know the Laplacian reduces to
\begin{eqnarray*}
L=\begin{bmatrix}
\alpha_2&-\frac{\sqrt{-2\alpha_2^2+2\alpha_2\alpha_3}}2&0\\-\frac{\sqrt{-2\alpha_2^2+2\alpha_2\alpha_3}}2&\alpha_3-\alpha_2&-\frac{\sqrt{-2\alpha_2^2+2\alpha_2\alpha_3}}2\\
0&-\frac{\sqrt{-2\alpha_2^2+2\alpha_2\alpha_3}}2&\alpha_2
\end{bmatrix}.
\end{eqnarray*}
\end{example}

\section{PST between end points of a path fails for the Laplacian}\label{sec:thmL}

Now we show that under XXX dynamics there is no weighted path on at least 3 vertices that admits PST between the end points.
As discussed  in Remark \ref{rmk:multpi}, if there is PST at time $\pi$ between the end vertices, the eigenvalues of the Laplacian are integers and alternate between even and odd (starting at even).
\begin{theorem}\label{thm:L}
No weighted (or unweighted) path on $n\geq 3$ vertices admits Laplacian PST between the end  points.
\end{theorem}

\begin{proof}
We begin by assuming $n$ is even. The persymmetric Laplacian is of the form
\begin{eqnarray}
\begin{bmatrix} \nonumber
r_1&-r_1&&&&&&&&\\
-r_1&r_1+r_2&-r_2&&&&&&&\\
&-r_2&r_2+r_3&-r_3&&&&&&\\
&&&\ddots&&&&&&\\
&&&-r_{\frac{n}2-1}&r_{\frac{n}2-1}+r_{\frac{n}2}&-r_{\frac{n}2}&&&&\\
&&&&-r_{\frac{n}2}&r_{\frac{n}2-1}+r_{\frac{n}2}&-r_{\frac{n}2-1}&&&\\
&&&&&&\ddots&&&\\
\\
&&&&&&&-r_2&r_1+r_2&-r_1\\
&&&&&&&&-r_1&r_1
\end{bmatrix} \\
\end{eqnarray}
with $B_1$ and $B_2$ written according to Lemma \ref{lem:CB}. Note the eigenvalues of $B_1$ are $\alpha_2,\cdots,\alpha_n$, and the eigenvalues of $B_2$ are $\alpha_1,\cdots,\alpha_{n-1}$.

First, we compute the determinant of $B_1$. Observe that $B_1$ can be written as $\hat{B_1} + 2r_{\frac{n}{2}}e_{\frac{n}{2}}e_{\frac{n}{2}}^T,$ where
$\hat{B_1}$ is the Laplacian matrix for a weighted path on $\frac{n}{2}$ vertices with edge weights $r_j, j=1, \ldots, \frac{n}{2}-1.$ We deduce that
$\det B_1 = \det \hat{B_1} + 2r_{\frac{n}{2}} c,$ where $c$ is the determinant of the leading principal submatrix of $\hat{B_1}$ of order
$\frac{n}{2}-1$. Evidently $\det \hat{B_1}=0,$ and applying the weighted matrix tree theorem \cite[Theorem 1.2]{wmtt}, we find that $c=r_1r_2 \ldots r_{\frac{n}{2}-1}$; hence  $\det B_1=2r_1r_2\cdots r_{\frac{n}2}$.
Thus we have
\begin{equation}\label{eq:B2rsneven}
2r_1r_2\cdots r_{\frac{n}2}=\alpha_2\alpha_4\cdots\alpha_n.
\end{equation}

Note that in this setting $B_2$  is a also a Laplacian matrix, and so the weighted matrix tree theorem  tells us that
all its cofactors of order $\frac{n}2-1$ are equal. The $(1, \frac n 2)$ cofactor in this case is $(-1)^{\frac n 2-1}(-r_1)(-r_2)\cdots(-r_{\frac n 2-1})=r_1r_2\cdots r_{\frac n 2-1}$.
Since the sum of all principal minors of size $\frac n 2-1$ equals the $(\frac n 2-1)$-th elementary symmetric function of
$\alpha_1,\alpha_3,\cdots,\alpha_{n-1}$, that is $\sum_{j=0,\cdots,\frac n 2 -1}\prod_{k\neq j}\alpha_{2k+1}$, and using the fact $\alpha_1=0$ (so the only nonzero product in the summand is $\alpha_3\cdots\alpha_{n-1}$), we have

\begin{equation}\label{eq:B1rsneven}
{\frac{n}2}r_1r_2\cdots r_{\frac{n}2-1}=\alpha_3\alpha_5\cdots\alpha_{n-1}.
\end{equation}

Combining equations (\ref{eq:B1rsneven}) and (\ref{eq:B2rsneven}), we find that
\begin{eqnarray}\label{eq:rforneven}
2r_{\frac n 2}&=&\frac{{\frac{n}2}\alpha_2\alpha_4\cdots\alpha_n}{\alpha_3\alpha_5\cdots\alpha_{n-1}}.
\end{eqnarray}
Now, $2r_{\frac{n}2}=S_1\in \Z$ by Corollary \ref{cor:middle}, and the numerator of the right hand side of equation (\ref{eq:rforneven}) is ${\frac{n}2}$ times all the odd eigenvalues while the denominator is the product of all the even eigenvalues. Thus we obtain a factor of $2^{{\frac{n}2}-1}$ in the denominator, from which it follows that $2^{{\frac{n}2}-1}$ divides ${\frac{n}2}$, which is a contradiction provided ${\frac{n}2}\geq 3$, i.e.\ provided $n\geq 6$.

For $n=4$, again from Corollary \ref{cor:middle},
we  find $\displaystyle r_2=\frac{S_1}2$ and  $\displaystyle q_2=\frac{S_2}{2S_1}$.
Substituting $q_2$ into the trace equations
\begin{eqnarray}\label{eq:trace}
\tr(B_1)=q_1+q_2+r_2&=&\alpha_2+\alpha_4 \label{eq:trace1}\\
\tr(B_2)=q_1+q_2-r_2&=&\alpha_1+\alpha_3\label{eq:trace2}
\end{eqnarray}
(with $q_2=r_1+r_2$) and then adding them gives $\displaystyle q_1=\frac{\alpha_2\alpha_4-\alpha_1\alpha_3}{S_1}$.
Now, substituting $q_1$ and $q_2$ into the determinant equations
\begin{eqnarray}
\det(B_1)=q_1q_2+q_1r_2-r_1^2&=&\alpha_2\alpha_4\label{eq:det1}\\
\det(B_2)=q_1q_2-q_1r_2-r_1^2&=&\alpha_1\alpha_3\label{eq:det2}
\end{eqnarray}
(with $q_1=r_1$) and then adding them gives
\begin{equation}
r_1=\sqrt{\frac{(\alpha_2\alpha_4-\alpha_1\alpha_3)S_2-(\alpha_2\alpha_4+\alpha_1\alpha_3)S_1^2}{2S_1^2}}.
\end{equation}
With the fact $\alpha_1=0$, we have $r_1=\sqrt{\frac{\alpha_2\alpha_4(S_2-S_1^2)}{2S_1^2}}$ and $q_1=\frac{\alpha_2\alpha_4}{S_1}$. Since $r_1=q_1$, we have $\alpha_2\alpha_4=\frac{S_2-S_1^2}{2}$, which tells us $2\alpha_2\alpha_4=\alpha_4\alpha_3-\alpha_3^2+\alpha_3\alpha_2$, with $\alpha_2$ and $\alpha_4$ being odd integers, and $\alpha_3$ an even integer.
Note that the left hand side of this equation is congruent to $2\modn 4$ while the right hand side is congruent to $0\modn 4$, and so we have a contradiction for $n=4$.
This completes the $n$ even case.

We now assume
 $n$ is odd. Our Hamiltonian is
 \begin{eqnarray}\nonumber
\begin{bmatrix}
r_1&-r_1&&&&&\\
-r_1&r_1+r_2&-r_2&&&&\\
&&\ddots&&&&&\\
&&-r_{\frac{n-1}2-1}&r_{\frac{n-1}2-1}+r_{\frac{n-1}2}&-r_{\frac{n-1}2}&&&\\
&&&-r_{\frac{n-1}2}&2r_{\frac{n-1}2}&-r_{\frac{n-1}2}&&&\\
&&&&-r_{\frac{n-1}2}&r_{\frac{n-1}2-1}+r_{\frac{n-1}2}&-r_{\frac{n-1}2-1}&\\
&&&&&\ddots&&\\
\\
&&&&&-r_2&r_1+r_2&-r_1\\
&&&&&&-r_1&r_1
\end{bmatrix}, \\
\end{eqnarray}
and again we take $B_1$ and $B_2$  as in Lemma \ref{lem:CB}.
 The eigenvalues of $B_1$ and $B_2$ interlace, with $\det B_1$ yielding $r_1r_2\cdots r_{\frac{n-1}2}=\alpha_2\alpha_4\cdots\alpha_{n-1}$ (the calculation is similar  to the $n$ even case), where $\alpha_2, \alpha_4, \dots, \alpha_{n-1}$ are the odd eigenvalues. For $B_2$,  its $(1,1)$ minor is just $\det(B_1)=r_1r_2\cdots r_{\frac{n-1}{2}}$. 
Now we calculate the other principal minors of size $(n-1)/2$ of $B_2$. Fix such a minor. If we take the factor $\sqrt{2}$ from the first row and the first column, then the principal minor that we seek is twice the principal minor of size $\frac{n-1}2$ of a Laplacian matrix; by the  weighted matrix tree theorem,
that minor is
 equal to the (1,1) minor of the Laplacian, which is $\det(B_1)$. Therefore the corresponding principal minors of $B_2$ are given by $2 \det(B_1)=2r_1\cdots r_{\frac{n-1}2}$.
Again, from the fact that the sum of all of $B_2$'s principal minors of size $(n-1)/2$ is equal to the $(\frac{n-1}{2})$-th elementary symmetric function of $\alpha_1,\alpha_3,\cdots,\alpha_{n+1}$, we find that $r_1\cdots r_{\frac{n-1}{2}}+\frac{n-1}{2} 2 r_1\cdots r_{\frac{n-1}{2}}=nr_1\cdots r_{\frac{n-1}2}=\alpha_3\alpha_5\cdots\alpha_n$. Combining this equation with the one for $B_1$, we have $n\alpha_2\alpha_4\cdots\alpha_{n-1}=\alpha_3\alpha_5\cdots\alpha_n$. This is a contradiction, since the left side of the equation is an odd number, while the right side is an even number.

This completes the $n$ odd case.
\end{proof}

We note that it was recently found \cite{trees} that there is no Laplacian PST for (unweighted) trees. Theorem \ref{thm:L} resolves the weighted generalization for the special case of paths. In fact, we can generalize the above theorem to any weighted  tree whose Laplacian matrix is persymmetric. Such a weighted tree can be represented schematically as follows (see Figure 1):

\begin{figure}[h]\label{Fig}
\begin{center}
\caption{Symmetric Trees}
 \begin{tikzpicture}[
        shorten >=1pt, auto, thick,
        node distance=3cm,
    main node/.style={circle,draw,font=\sffamily\Large\bfseries, every loop/.style={}}
                            ]
    \node[main node] (1) {$G$};
    \node[main node] (2) [right of= 1]{$\tilde{G}$};
    \path[every node/.style={font=\sffamily\small}]
        (1) edge  node {$w_1$} (2);
\end{tikzpicture}
\\
or
\\\vspace{0.2cm}

\begin{tikzpicture}[
        shorten >=1pt, auto, thick,
        node distance=3cm,
    main node/.style={circle,draw,font=\sffamily\Large\bfseries, every loop/.style={}}
                            ]
    \node[main node] (1) {$G$};
    \node[main node] (2) [right of= 1]{$v$};
    \node[main node] (3) [right of= 2]{$\tilde{G}$};
    \path[every node/.style={font=\sffamily\small}]
        (1) edge  node {$w_1$} (2)
         (2)   edge node {$w_1$} (3);
\end{tikzpicture}
\\
\end{center}
\end{figure}

\noindent where $G$ is a weighted tree, $\tilde{G}$ is the mirror image of $G$, and $w_1$  is an edge weight. The first graph (with a weighted edge connecting a vertex in $G$ to its corresponding vertex in $\tilde{G}$) generalizes weighted paths with $n$ even, while the second graph (with one middle vertex $v$ connected to a vertex in $G$ and to the corresponding vertex in $\tilde{G}$) generalizes weighted paths with $n$ odd.

Lemma \ref{lem:CB} applies equally well in the situation of a weighted  tree whose Laplacian matrix is persymmetric. For such a weighted tree, the matrix $C$ is as in Corollary \ref{cor:middle}: it has one non-zero entry for $n$ even (so the two matrices $B_1$ and $B_2$ are still rank one Hermitian perturbation of each other) and it is the zero matrix for $n$ odd ($B_2$ is a bordered Hermitian matrix of $B_1$). Although the Hamiltonian (and thus $B_1$ and $B_2$) is more complicated than in  Corollary \ref{cor:middle}, the interlacing of the eigenvalues of $B_1$ and $B_2$ still holds, and the arguments using the weighted matrix tree theorem continue to apply.

Assume $n$ is even, now observe that if $v$ is an eigenvector of $B_1=E-RC$ associated to the eigenvalue $\lambda$, i.e., $(E-RC)v=\lambda v$, then from
$\begin{bmatrix}E&RCR\\C&RER\end{bmatrix}\begin{bmatrix}v\\-Rv\end{bmatrix}=\begin{bmatrix}Ev-RCv\\Cv-REv\end{bmatrix}=\lambda\begin{bmatrix}v\\-Rv\end{bmatrix}$, we know the antisymmetric vector $\begin{bmatrix}v\\-Rv\end{bmatrix}$ is an eigenvector of the Hamiltonian $L=\begin{bmatrix}E&RCR\\C&RER\end{bmatrix}$ associated to the eigenvalue $\lambda$. Similarly, if $u$ is the eigenvector of $B_2=E+RC$ associated to the eigenvalue $\mu$, then the symmetric vector $\begin{bmatrix}u\\Ru\end{bmatrix}$ is an eigenvector of $L$ associated to eigenvalue $\mu$.
Using all the $n/2$ orthogonal eigenvectors $v_j$ of $B_1$ and the $n/2$ orthogonal eigenvectors $u_j$ of $B_2$, we can form $n$ orthogonal eigenvectors of $L$: $\begin{bmatrix}v_j\\-Rv_j\end{bmatrix}$, $\begin{bmatrix}u_j\\Ru_j\end{bmatrix}$, $j=1,\ldots, n/2$. Normalize each of them and use them as columns to form a real orthogonal matrix $S$; assume it diagonalizes $L$ to $\Lambda$.
If there is PST between a vertex $j$ and its mirror image $n+1-j$, then $s_j e^{i\pi\Lambda}=e^{i\phi}s_{n+1-j}$, where $s_{\ell}$ is the $\ell$-th row of $S$ and $\phi$ is some real number.
If we assume $S$ does not have any zero entries, then
from the symmetric and antisymmetric structures of the eigenvectors, and the fact that  0 is an eigenvalue, we know the eigenvalues of $B_1$ are odd integers, and the eigenvalues of $B_2$ are even integers.
So the arguments in Theorem \ref{thm:L} applies if each of the two matrices $B_1$ and $B_2$ formed from the Laplacian of a persymmetric weighted tree can be diagonalized by some real orthogonal matrix which does not have zero entries.

If $n$ is odd, then the Laplacian of the tree is $L=\begin{bmatrix}
E&x&0\\x^T&2w_1&x^TR\\0&Rx&RER\end{bmatrix}$, where $x=\begin{bmatrix}0\cdots 0\, -w_1\end{bmatrix}^T\in\mathbb{R}^{\frac{n-1}2}$. As above, we can check if $v$  is an eigenvector of $B_1=E-RC=E$ associated to the eigenvalue $\lambda$, i.e., $Ev=\lambda v$, then $\begin{bmatrix}v\\0\\-Rv\end{bmatrix}$ is an eigenvector of $L$ associated to the eigenvalue $\lambda$.
And if $u=\begin{bmatrix}a\\ \tilde{u}\end{bmatrix}$ is an eigenvector of $B_2=\begin{bmatrix}2w_1&\sqrt{2}x^T\\\sqrt{2}x&E\end{bmatrix}$ associated to the eigenvalue $\mu$, then $\begin{bmatrix}\tilde{u}\\ \sqrt{2}a\\R\tilde{u}\end{bmatrix}$ is an eigenvector of $L$ associated to the eigenvalue $\mu$. Using the $\frac{n-1}2$ eigenvectors $v_j$ of $B_1=E$ and the $\frac{n+1}2$ eigenvectors $u_j$ of $B_2$, we form $n$ orthogonal eignvectors of $L$: $\begin{bmatrix}v_j\\0\\-Rv_j\end{bmatrix}$, $\begin{bmatrix}\tilde{u}_j\\\sqrt{2}a\\R\tilde{u}_j\end{bmatrix}$.
Now assume that $B_1$ can be diagonalized by a real orthogonal matrix which does not have zero entries,  and that $B_2$ can be diagonalized by a real orthogonal matrix which does not have zero entries apart from the first row. With a similar argument as in the $n$  even case, we can see if  there is PST between a vertex $j<\frac{n+1}2$ and its mirror image $n+1-j$, then the eigenvalues of $B_1=E$ are odd integers, and the eigenvalues of $B_2$ are even integers. Therefore the arguments in Theorem \ref{thm:L} applies here.


We summarize the above arguments in the following Theorem:

\begin{theorem} Let $G$ be a symmetric tree as in Figure 1 (namely, a weighted tree whose Laplacian matrix is persymmetric). Suppose the two matrices $B_1$ and $B_2$ we obtain from the Laplacian $L$   can be diagonalized by some real symmetric matrices $Q_1$ and $Q_2$, respectively, such that $Q_1$ contains no zero entries, and $Q_2$ contains no zero entries if $n$ is even and contains no zero entries apart from the first row if $n$ is odd. Then $G$ does not admit Laplacian PST between vertex $\displaystyle j<\frac{n+1}2$ and its mirror image $n+1-j$.
\end{theorem}

Symmetric trees are special cases of graphs with an involution; PST and PGST properties of such graphs were studied in \cite{KemptonPGST} under XX dynamics.  The Hamiltonian considered in \cite{KemptonPGST} was the adjacency matrix plus a diagonal matrix rather than the weighted Laplacian matrix considered here.  The results in \cite{KemptonPGST} and those in this section are independent in that neither implies the other.

\section{Adjacency Matrices and the Rational weights conjecture}\label{sec:conjrat}

Henceforth, we focus on XX dynamics (and the corresponding adjacency matrix).  We give a complete analysis of the $4\times 4$ and $5\times 5$ cases and prove a more general result motivated by an observation made in the $4\times 4$ case.

\begin{example}
For a weighted path (with no loops) on 4 vertices, using a similar computation to that in the proof of Theorem \ref{thm:L} for the $n=4$ case,
and with eigenvalues $\alpha_1, \dots, \alpha_4$ of the adjacency matrix written as  $-\beta_2, -\beta_1$, $\beta_1, \beta_2$, we have $r_2=\beta_2-\beta_1$
and $r_1=\sqrt{\beta_1\beta_2}$. The Hamiltonian is then
\begin{eqnarray}
A=\begin{bmatrix}
0&\sqrt{\beta_1\beta_2}&0&0\\
\sqrt{\beta_1\beta_2}&0&\beta_2-\beta_1&0\\
0&\beta_2-\beta_1&0&\sqrt{\beta_1\beta_2}\\
0&0&\sqrt{\beta_1\beta_2}&0\\
\end{bmatrix}.
\end{eqnarray}

It is clear from the above matrix why PST does not occur in the unweighted path on 4 vertices. In this case $r_1=r_2=1$ and therefore  $\beta_1=\frac{-1+\sqrt{5}}{2}$ and $\beta_2=\frac{1+\sqrt{5}}{2}$. Since $\beta_2/\beta_1$ is irrational, there is no nonzero constant $\kappa$ for which both $\kappa \beta_1$ and $\kappa \beta_2$ are integers. Hence the unweighted path on four vertices cannot have PST (this is shown more generally for unweighted paths of length four or greater with loops in \cite{KemptonPST}).  More generally no weighted path without potentials on four vertices with all rational weights can have PST.  By Remark \ref{rmk:multpi2}, we set $\beta_1=1\modn 4$ and $\beta_2=3\modn 4$ without loss of generality (we scale the Hamiltonian by a factor 2 to have integer eigenvalues). It follows that $\beta_2-\beta_1$ is an even integer. However, $\beta_1\beta_2\equiv 3 \modn4$, so the quantity is not a perfect square, and therefore $\sqrt{\beta_1\beta_2}$ is irrational.  There is no nonzero constant $\kappa$ for which both $\kappa (\beta_2-\beta_1)$ and $\kappa \sqrt{\beta_1\beta_2}$ are rational.  This observation motivates a more general result which we will present after we analyse the $5\times 5$ case.

\end{example}

\begin{example}[$5\times 5$ Case]
For $n=5$, we can solve for $q_3$ in terms of the eigenvalues as before.  However the trace equation now has two unknowns $q_1$ and $q_2$, so we cannot use our previous method to solve for these entries.  The case with no potentials is still amenable.  The eigenvalues of the adjacency matrix in this case are $-\beta_2<-\beta_1<\beta_0=0<\beta_1<\beta_2$ and the two polynomials are $p_5(x)=x(x^2-\beta_2^2)(x^2-\beta_1^2)=x^5-(\beta_2^2+\beta_1^2)x^3+\beta_1^2\beta_2^2x$ and $\tilde{p}_4(x)=b_4x^4+b_2x^2+b_0$ for some real numbers $b_0, b_2, b_4$. The system of equations to solve is
\begin{eqnarray}
b_4\beta_2^4+b_2\beta_2^2+b_0=1\\
b_4\beta_1^2+b_2\beta_1^2+b_0=-1\\
b_0=1
\end{eqnarray}
which has the corresponding solutions $\displaystyle b_4= \frac{2}{\beta_1^2(\beta_2^2-\beta_1^2)}$, and $\displaystyle b_2=\frac{-2\beta_2^2}{\beta_1^2(\beta_2^2-\beta_1^2)}$, $\displaystyle b_0=1
$. We now have $\tilde{p}_4(x)=\frac{2[x^4-\beta_2^2x^2+\frac{1}{2}\beta_1^2(\beta_2^2-\beta_1^2)]}{\beta_1^2(\beta_2^2-\beta_1^2)}$ with the monic version being $p_4(x)=x^4-\beta_2^2x^2+\frac{1}{2}\beta_1^2(\beta_2^2-\beta_1^2)$. Now performing the subtraction $p_5-xp_4$ yields $r_1^2=\beta_1^2$ and the new monic polynomial $p_3(x)=x^3-\frac{1}{2}(\beta_1^2+\beta_2^2)x$. Repeating this again with $p_4$ and $p_3$ gives $r_2^2=\frac{1}{2}(\beta_2^2-\beta_1^2)$. The Hamiltonian is now
\begin{eqnarray}
A=\begin{bmatrix}
0&\beta_1&0&0&0\\
\beta_1&0&\sqrt{\frac{\beta_2^2-\beta_1^2}{2}}&0&0\\
0&\sqrt{\frac{\beta_2^2-\beta_1^2}{2}}&0&\sqrt{\frac{\beta_2^2-\beta_1^2}{2}}&0\\
0&0&\sqrt{\frac{\beta_2^2-\beta_1^2}{2}}&0&\beta_1\\
0&0&0&\beta_1&0
\end{bmatrix}.
\end{eqnarray}
\end{example}

Theorem \ref{thm:L} tells us that no weighted path of length at least 3 has Laplacian PST between its end vertices. Contrast this with the adjacency matrix setting, where there is a weighted path (with no loops) of any length that admits PST between
its end vertices.   We have a conjecture about the weights: if all the weights of a weighted path on at least 4 vertices are rational numbers, then 
there is no adjacency matrix PST at time $\pi$ between the end vertices of the path. We confirm that conjecture in the cases that $n=4$, $n\equiv 5\modn 8$ and for $n\equiv 3\modn 8$ but $n\neq 3$.

\begin{proposition}\label{prop:rat}
Suppose that $n=4,$ or $n\geq 5$ and $n\equiv 3\modn 8$ or $n\equiv 5\modn 8$.
If the weights of a weighted path on $n$ vertices with or without potentials are all rational numbers, 
then there is no adjacency matrix PST between its end vertices at readout time $\pi$.
\end{proposition}

\begin{proof}
As discussed, for a weighted path that exhibits PST at time $\pi$  between its end vertices, 
by performing an overall energy shift if necessary (which does not change the PST readout time), we can make all its eigenvalues integers (in particular, with the smallest one being an odd integer) with alternating parity. See Remark \ref{rmk:multpi}. 

 For $n=4$, from Corollary \ref{cor:middle},  $r_2=\frac{S_1}{2}\in \mathbb{Q}$. We show that $r_1=\sqrt{\frac{(\alpha_2\alpha_4-\alpha_1\alpha_3)S_2-(\alpha_2\alpha_4+\alpha_1\alpha_3)S_1^2}{2S_1^2}}$ (calculated in Theorem \ref{sec:thmL}) cannot be rational by showing that $\displaystyle \frac{(\alpha_2\alpha_4-\alpha_1\alpha_3)S_2-(\alpha_2\alpha_4+\alpha_1\alpha_3)S_1^2}2$ is not a perfect square. Rearranging the terms we have
\begin{eqnarray}\label{eq:r1}
\begin{split}
&\frac12(\alpha_2\alpha_4(S_2-S_1^2)-\alpha_1\alpha_3(S_2+S_1^2))\\
&\quad=-\alpha_2\alpha_4(\alpha_1^2+\alpha_3^2)-\alpha_1\alpha_3(\alpha_2^2+\alpha_4^2)+\alpha_2\alpha_4
\\
&\quad \times (\alpha_1\alpha_2-\alpha_1\alpha_3+\alpha_1\alpha_4+\alpha_2\alpha_3-\alpha_2\alpha_4+\alpha_3\alpha_4)\\
&\quad -\alpha_1\alpha_3\alpha_2\alpha_4+\alpha_1\alpha_3(\alpha_1+\alpha_3)(\alpha_2+\alpha_4)-\alpha_1^2\alpha_3^2
\end{split}
\end{eqnarray}

From the fact that $\alpha_1 $ and $\alpha_3$ are odd integers, and $\alpha_2$ and $\alpha_4$ are even integers, we know the first 5 terms in the summand are all divisible by 4, and therefore their sum is congruent to either $0\modn 8$ or $4 \modn 8$. Since the square $\alpha_1^2\alpha_3^2$ of an odd integer $\alpha_1\alpha_3$ is congruent to $1 \modn 8,$ we know the result in equation (\ref{eq:r1}) is congruent to either $ 3\modn 8$ or $ 7\modn 8$, and hence is not a perfect square. Thus $r_1$ is not rational, which establishes the result for $n=4$.  In fact we can say more: if the weights of a weighted path on 4 vertices are all rational numbers, then there is no adjacency matrix PST between its end vertices at any time, since we can not scale the adjacency matrix such that $r_1$ and $r_2$ are both rational.
\medskip

\medskip

Next, suppose that $n\geq 5$ and $n\equiv 3\modn 8$ or $n\equiv 5\modn 8$. Observe that
 since $n$ is odd, by Corollary \ref{cor:middle}  we have $\displaystyle r_{\frac{n-1}2}=\frac{\sqrt{S_2-S_1^2}}2$. 
We claim now that the quantity $S_2-S_1^2$ is not a perfect square, and so $r_\frac{n-1}{2}$ is irrational. 
To see the claim, note that $S_2-S_1^2=-2[\sum_{r}\alpha_{2r}^2+\sum_{1\leq j<k\leq n}(-1)^{j+k}\alpha_j\alpha_k]$, where $\sum_r \alpha_{2r}^2$ is divisible by 4. Consequently, $2\sum_{r}\alpha_{2r}^2\equiv 0 \modn 8$.

If we can show that $\sum_{1\leq j<k\leq n}(-1)^{j+k}\alpha_j\alpha_k$ is odd, then we can conclude that $S_2-S_1^2$ is not a perfect square. To this end, it is enough to count the number of distinct pairs of odd numbers appearing in the summation.  If $n=2m-1$ for some $m\in \Z$, then $m$ is the number of odd numbers in the sum, and the number of distinct odd pairs is $m(m-1)/2$. For $n\equiv 3\modn 8$ or $n\equiv 5\modn 8,$ we have $m\equiv 2\modn 4$ or $m\equiv 3\modn 4$, respectively. In either case, $m(m-1)/2$ is odd and the claim follows.  Thus $S_2-S_1^2$ is not a perfect square, so $r_\frac{n-1}{2}$ is not rational.

\end{proof}

\section{Acknowledgements}

 S.K., R.P., and S.P. were supported by NSERC Discovery Grants RGPIN/6123-2014, 400550, and 1174582, respectively;  S.P. is also supported by the Canada Foundation for Innovation and the Canada Research Chairs Program.
D.M.\ was supported through a NSERC Undergraduate Student Research Award; X.Z.\ was supported by the University of Manitoba's Faculty of Science and Faculty of Graduate Studies.

\end{document}